\documentclass[a4paper,conference]{IEEEtran}
\usepackage{xcolor}
\usepackage{verbatim,amsmath,enumerate,amssymb,hhline}

\def\norm#1{\left\| #1 \right\|}
\newtheorem{definition}{Definition}[section]
\newtheorem{thm}{Theorem}[section]
\newtheorem{proposition}[thm]{Proposition}
\newtheorem{lemma}[thm]{Lemma}

\newtheorem{exam}{Example}[section]

\def\bea{\begin{IEEEeqnarray}{rCl}} 
\def\eea{\end{IEEEeqnarray}}
\def\bean{\begin{IEEEeqnarray*}{rCl}} 
\def\eean{\end{IEEEeqnarray*}}

\newtheorem{remark}{Remark}[section]

\DeclareMathOperator*{\diag}{diag}

\providecommand{\abs}[1]{\ensuremath{\left\lvert #1 \right\rvert}}
\providecommand{\norm}[1]{\ensuremath{\left\Vert #1 \right\Vert}}

\providecommand{\vv}[1]{\textquotedblleft #1\textquotedblright}
\newcommand{\Q}{\mathbb{Q}}
\newcommand{\Z}{\mathbb{Z}}
\newcommand{\C}{\mathbb{C}}
\newcommand{\R}{\mathbb{R}}

\DeclareMathOperator*{\Vol}{Vol}
\providecommand{\abs}[1]{\ensuremath{\left\lvert #1 \right\rvert}}
\providecommand{\norm}[1]{\ensuremath{\left\Vert #1 \right\Vert}}

\providecommand{\vv}[1]{\textquotedblleft #1\textquotedblright}

\newcommand{\D}{{\mathcal D}}

\newcommand{\OO}{{\mathcal O}}

\newcommand{\A}{{\mathcal A}}

\newcommand{\mindet}[1]{\hbox{\rm det}_{min}\left( #1\right)}

\usepackage{eufrak}

\providecommand{\vv}[1]{\textquotedblleft #1\textquotedblright}

\begin{document}

\title{Estimates for  the growth of inverse determinant sums of quasi-orthogonal and number field lattices }
\author{
\IEEEauthorblockN{Roope Vehkalahti}
\IEEEauthorblockA{Department of mathematics,
University of Turku\\
Finland\\
Email: roiive@utu.fi}
\and
\IEEEauthorblockN{Laura Luzzi}
\IEEEauthorblockA{Laboratoire ETIS,
 CNRS - ENSEA - UCP \\
Cergy-Pontoise, France \\
laura.luzzi@ensea.fr}

}

\maketitle

\begin{abstract}
Inverse determinant sums appear naturally as a tool for analyzing performance of space-time codes in Rayleigh fading channels.
This work will analyze the growth of inverse determinant sums of a family of quasi-orthogonal codes and will show that the growths are in logarithmic class.  This is  considerably lower than that of  comparable number field codes.  
\end{abstract}


\section{Introduction}

In \cite{VLL2013} inverse determinant sums were proposed as a tool to
analyze the performance of algebraic space-time codes for MIMO fading
channels. These sums can be seen as a generalization of the theta
series for the Gaussian channel. They arise naturally from the union bound on
the pairwise error probability for spherical constellations, but also in the analysis of fading wiretap channels \cite{BO}.

In  \cite{VLL2013} the authors analyzed the growth of the  inverse determinant sums of  diagonal number field  codes and  of most well known division algebra codes. In this work we are going to extend the analysis to  a large class of quasi-orthogonal codes. Our work will reveal that the growth of 
inverse determinant sums of the analyzed codes is considerably smaller than that of the corresponding diagonal number field codes. This difference suggest that asymptotically, with growing constellation, quasi-orthogonal codes are  considerably better than number field codes. This difference can not be captured in the framework of diversity-multiplexing gain tradeoff. 

For related work, we refer the reader to \cite{WZ} and \cite{JHL}.

\begin{remark}
This is a corrected and extended version of \cite{IZS}. The main correction is to Theorem \ref{Loxton} which we quoted carelessly in the previous version. This correction does affect all the formulas that are derived from it. However, it does not change any formulas related to complex constellations. In the case of totally real fields the change sometimes doubles  the number of units. The other correction is to the definition of the normalized inverse determinant sum \eqref{normalized}. This correction does not change any results, but is corrected for consistency.

The extensions in this paper are Propositions \ref{complexsingle} and \ref{realsingle}. These results follow from improving on the estimate for the truncated Dedekind zeta function done in \cite{VLL2013}. The  new estimate is Lemma \ref{tom}. The proof of this result is due to Tom Meurman and the authors are  grateful to him.
\end{remark}
\begin{remark}
We note that after the publication of \cite{IZS} the authors in  \cite{dave}  extended Proposition \ref{numberfieldbound} to ideals and  proved that our upper bound is tight in the case of principal ideal domains.
\end{remark}

\section{Inverse determinant sum}

We begin by providing basic definitions concerning matrix lattices and spherical constellations, that are needed in the sequel.  

\subsection{Matrix lattices and spherically shaped coding schemes}\label{latticesection}

\begin{definition}
A {\em space-time lattice code} $C \subseteq M_n(\C)$ has the form
$$
\Z B_1\oplus \Z B_2\oplus \cdots \oplus \Z B_k,
$$
where the matrices $B_1,\dots, B_k$ are linearly independent over $\R$, \emph{i.e.}, form a lattice basis, and $k$ is
called the \emph{rank}  or the \emph{dimension} of the lattice.
\end{definition}

\begin{definition}\label{def:NVD}
If the minimum determinant of the lattice $L \subseteq M_n(\C)$ is non-zero, i.e. it satisfies
\[
\inf_{{\bf 0} \neq X \in L} \abs{\det (X)} > 0, 
\]
we say that the code has a \emph{non-vanishing determinant} (NVD).
\end{definition}

We now consider a coding scheme based on a $k$-dimensional lattice $L$ inside $M_{n}(\C)$. For a given positive real number $M$ we define the finite code 
$$
L(M)=\{a  \,|\,a \in L, \norm{a}_F \leq M \},
$$
where $\norm{a}_F$ refers to the Frobenius norm.
In the following we will also use the notation
$$
B(M)=\{a  \,|\,a \in M_n(\C),\; \norm{a}_F \leq M \},
$$
for the sphere with radius $M$.

Let $L\subseteq M_n(\C)$ be a $k$-dimensional lattice. For any fixed $m\in \Z^+$ we define 
$$
S_L^m(M):=\sum_{X\in L(M) \setminus \{\mathbf{0}\}} \frac{1}{|\det(X)|^{m}}.
$$
Our main goal is to study the growth of this sum as $M$ increases. Note, however, that in order to have a fair comparison between two different space-time codes, these should be normalized to have the same average energy. Namely, the volume $\Vol(L)$ of the fundamental parallelotope 
$$\mathcal{P}(L)=\{\alpha_1 B_1 + \alpha_2 B_2 + \ldots + \alpha_k B_k \;|\; \alpha_i \in [0,1) \;\; \forall i\}$$
should be normalized to $1$. The normalized version of the inverse determinant sums problem is then to consider the growth of the sum $\tilde{S}_L^m(M)=S_{\tilde{L}}^m(M)$ over the lattice $\tilde{L}=\Vol(L)^{-1/k}L$. Since $\tilde{L}(M)=\Vol(L)^{-1/k}L(M \Vol(L)^{1/k})$, we have 
\begin{equation}
\tilde{S}_L^m(M)= \Vol(L)^{mn/k} S^m_L(M \Vol(L)^{1/k}).
\label{normalized}
\end{equation}

\begin{remark}
In \cite{IZS} we erroneously stated that 
$$
\tilde{S}_L^m(M) =\Vol(L)^{mn/k} S^m_L(M).
$$
The correct form  is  \eqref{normalized}, which was originally stated in \cite{LV}. The difference is that in the correct expression   the radius of the ball also increases. However, this difference does not affect the results in this paper. In every case in this paper the growth of the inverse determinant is logarithmic and therefore the constant term $\Vol(L)^{1/k}$ will have no effect on the dominating term of the sum.
\end{remark}

\subsection{Inverse determinant sums and error performance of space-time lattice codes}
Let us now consider the slow Rayleigh fading MIMO channel with $n$ transmit and $n_r$  receive antennas. The channel equation can then be written as
\[
Y \ = \ H X + N,
\]
where $H$ and $N$ are respectively the channel and noise matrices. We suppose that the transmitted codeword $X$ belongs to a finite code $L(M) \subset M_n(\C)$ carved from a $k$-dimensional NVD lattice $L$ as defined previously. In terms of pairwise error probability, we have for $X \neq X'$
\[
P(X \to X')\leq \frac{1}{| \det(X-X')|^{2n_r}}, 
\]
and the corresponding upper bound on overall error probability 
\[
P_e \ \leq \ \sum_{X\in L,\, 0 < \norm{X}_F\leq 2M}\frac{1}{ |\det(X)|^{2n_r}}.
\]

Our main goal is now to study the growth of the  sum $S_L^m(M)$  as $M$ increases. 
In particular, we want to find, if possible, a function $f(M)$ such that
$$
 S_L^m(M) \sim f(M).
$$

\section{Inverse determinant sums of algebraic number field codes} \label{sec:4}

In this section we will give and review some results concerning inverse determinant sums of diagonal number field codes.
These results will play an important role in our analysis of quasi-orthogonal codes. The proofs are  mostly analogous to those given in \cite{VLL2013} and we will skip them. However, the case where the receiver has a single antenna will improve on the original one given in \cite{VLL2013}.  Unlike in the rest of the paper we will  state the results in the normalized from  $\tilde{S}_L^m(M)$ following the general normalization given in   \cite{LV}.

We need  first some preliminary results. Let us suppose we have an algebraic number field $K$. We will denote with $\OO_K$ the ring of algebraic integers in $K$. If $B$ is an integral ideal of $K$ we will use  the notation $N(B)=[\OO_K:B]$.

Let $I_{\OO_K}$ be the set of nonzero ideals of the ring $\OO_K$. The \emph{Dedekind zeta function} of the number field $K$ is 
\begin{equation}\label{Dedekind_zeta_function}
\zeta_{K}(s)=\sum_{A\in I_{\OO_K}}\frac{1}{N(A)^s},
\end{equation}
where  $s$ is a complex number with  $\Re(s)>1$.
For a truncated version of this sum we will use notation 
$$
\zeta_K(s,M):=\sum_{A\in I_{\OO_K},  N(A)\leq M}\frac{1}{N(A)^s}.
$$
In the following  we will need an estimate for $\zeta_K(s,M)$. We will begin with a classical result. Let us suppose we have an algebraic number field
$K$ with \emph{signature} $(r_1, r_2)$.
\begin{thm}
The function $\zeta_K(s)$ converges for all $s$, $\Re(s)>1$ and has a simple pole at $1$ with residue
$$
\lim_{s\rightarrow 1} (s-1) \zeta_K(s)= \frac{2^{r_1}(2\pi)^{r_2}h_K R_K}{\omega\sqrt{|d(K)|}},
$$
where $R_K$ is the regulator of the number field $K$, $d(K)$ the discriminant and $\omega$ the number of roots of unity in $K$.
\end{thm}

In the following we will use the shorthand
$$
\alpha_K=\frac{2^{r_1}(2\pi)^{r_2} R_K}{\omega\sqrt{|d(K)|}}.
$$

In order to estimate the truncated Dedekind zeta function at the point 1 we need  an estimate for the number of ideals with bounded norm. Let us denote with  $N(K, M)$ the number of integral  ideals  of $K$ with norm less than or equal to $M$. We then have the following.

\begin{proposition}\cite[Theorem 5]{MO}\label{vano}
Given a degree $n$ number field $K$ there exist  positive constants $c$ and  $a$, independent of $M$, such that
$$
|N(K,M)-\alpha_K h_K M| \leq c M^{1-a}.
$$
\end{proposition}
We are now ready to estimate the truncated Dedekind zeta function.

\begin{lemma}\label{tom}
We have the following
$$
\sum_{A\in I_{\OO_K},  N(A)\leq M}\frac{1}{N(A)}= h_K\alpha_K \log{M}+O(1).
$$
\end{lemma}
\begin{proof}
Let us first write
$$
\sum_{  N(A)\leq M}\frac{1}{N(A)}=\sum_{n\leq M}\frac{1}{n}\sum_{N(A)=n} 1=\sum_{n\leq M} \frac{z(n)}{n}.
$$
Manipulating  the previous expression, we can write it  in  the form
\begin{equation}\label{lavennus}
S+T=\sum_{n\leq M} \frac{z(n)-h_K\alpha_K }{n}+\sum_{n\leq M} \frac{h_K\alpha_K}{n}.
\end{equation}
Let us now analyze the first term  $S$ by using partial summation  \cite[A. 21]{ivic}
$$
S=\frac{1}{M}\sum_{n\leq M} (z(n)-h_K\alpha_K)+\int_1^M \frac{1}{t^2}\sum_{n\leq t} (z(n)-h_K\alpha_K)dt.
$$
According to Proposition \ref{vano} we have that
$$
\sum_{n\leq t} (z(n)-h_K\alpha_K)=\sum_{n\leq t} z(n) - h_K\alpha_K t+ O(1)=O(t^{1-a}).
$$
It follows that
 $$
\int_1^M \frac{1}{t^2}\sum_{n\leq t}(z(n)-h_K\alpha_K)dt=\int_1^M  O(t^{-1-a})dt.
$$

Putting now all together we have that
$$
S=\frac{1}{M}O(M^{1-a})+\int_1^M O(t^{-1-a})dt=O(M^{-a})+O(1).
$$
Estimating  $\sum_{n\leq M} \frac{h_K\alpha_K}{n}=h_K\alpha_K \log{M}+O(1)$,  gives us the final result.
\end{proof}

\begin{remark}
This results improves on the very crude estimate  \cite{VLL2013}
$$
\sum_{A\in I_{\OO_K},  N(A)\leq M}\frac{1}{N(A)}\leq d(\log{M})^n,
$$
where $d$ is some constant.
\end{remark}

\subsection{Inverse determinant sums of real  diagonal number field codes}
Let $K$ be a totally real number field of degree $n$ and let $\{ \sigma_1, \cdots, \sigma_n \}$ be the $\Q$-embeddings   from $K$ to $\R$. 
We then have the canonical embedding $\psi:K\mapsto M_n(\R)$ defined by
 $$
\psi(x)=\diag(\sigma_1(x),\dots, \sigma_n(x)).
$$
It is a well known result that $\psi(\OO_K)$ is an $n$-dimensional NVD lattice in $M_n(\R)$. Let us now consider  the corresponding inverse determinant sum.  The main role in the analysis is played by the following unit group density result.

\begin{thm}[\cite{EL}]\label{Loxton}
Let us suppose that $[K:\Q]=n$, we then have that
$$
|\psi(\OO_K^*)\cap B(M)|= N_K(\log{M})^{n-1} + O((\log{M})^{n-2}),
$$
where $N_K = \frac{\omega n^{n-1}}{R (n-1)!}$.
\end{thm}
Here $R$ is the regulator of the number field $K$, $\omega$ the number of roots of unity in $K$.

We first have a real analogue to the results in \cite{VLL2013}.

\begin{proposition}\label{numberfieldbound}
Let us suppose that $K$ is a totally real number field with $[K:\Q]=n$ and that $m>1$. Then 
$$
\tilde{S}_{\psi(\OO_K)}^{m}(M)\leq  \tilde{N_K} \zeta_K(m) \left(\log{M}\right)^{n-1} + O((\log{M})^{n-2}) 
$$
and
$$
\tilde{N}_K\left(\log{M}\right)^{n-1} + O((\log{M})^{n-2})\leq \tilde{S}_{\psi(\OO_K)}^{m}(M),
$$
where $\tilde{N}_K = \frac{\omega(n)^{n-1}}{R_K (n-1)!}(\sqrt{|d(K)|})^{m}$. 
\end{proposition}

In the case of single receive antenna we  now have the following.
\begin{proposition}\label{realsingle}
Let us suppose that $K$ is a totally real number field with $[K:\Q]=n$ and that $m=1$. Then 
$$
\tilde{S}_{\psi(\OO_K)}^{m}(M)\leq  c_K\left(\log{M}\right)^{n} + O((\log{M})^{n-1}),
$$
where $c_K=\frac{h_K2^n n^{n}}{(n-1)!}$.
\end{proposition}
\begin{proof}
Following the proof of Proposition 4.4 and Section 4.C in \cite{VLL2013} we have an estimate
$$
\tilde{S}_{\psi(\OO_K)}^{m}(M)\leq \zeta_K(1,\frac{M^n}{n^{n/2}})(\tilde{N}_K(\log{M})^{n-1} + O((\log{M})^{n-2})),
$$
and the end result now follows from plugging  the result of Lemma \ref{tom} and Theorem \ref{Loxton} to this  estimate.
\end{proof}

For principal ideal domains the upper bound is likely tight.

\subsection{Inverse determinant sums of complex diagonal number field codes}
 
Let $K/\Q$ be  a totally complex extension of degree $2n$ and $\{\sigma_1,\dots,\sigma_n\}$ be a set  of  $\Q$-embeddings, such that we have chosen one from each complex conjugate pair. Then we can define a
\emph{relative canonical embedding} of $K$ into $M_n(\C)$ by
$$
\psi(x)=diag(\sigma_1(x),\dots, \sigma_n(x)).
$$
The ring of algebraic integers $\OO_K$ has a  $\Z$-basis $W=\{w_1,\dots ,w_{2n}\}$ and $\psi(W)$ is a $\Z$-basis for the full  lattice $\psi(\OO_K)$ in $M_n(\C)$.

\begin{proposition}\cite[Section 4.C]{VLL2013}\label{complexbound}
Let $K$ be a totally complex algebraic number field of degree $2n$. If $n_r>1$,  we  have that
$$
\tilde{S}_{\psi(\OO_K)}^{2n_r}(M)\leq  \tilde{N}_K \zeta_K(n_r) \left(\log{M}\right)^{n-1} + O((\log{M})^{n-2}) 
$$
and
$$
\tilde{N}_K\left(\log{M}\right)^{n-1} + O((\log{M})^{n-2})\leq \tilde{S}_{\psi(\OO_K)}^{2n_r}(M),
$$
where $\tilde{N}_K = \frac{ \omega (n)^{n-1}}{R (n-1)!}(2^{-n}\sqrt{|d(K)|})^{n_r}$. 
\end{proposition}

\begin{proposition}\label{complexsingle}
Let $K$ be a totally complex algebraic number field of degree $2n$. We then have
\begin{equation}\label{upper}
\tilde{S}_{\psi(\OO_K)}^{2}(M)\leq  c_K \left(\log{M}\right)^{n} + O((\log{M})^{n-1})
\end{equation}
and
\begin{equation}
\tilde{N}_K\left(\log{M}\right)^{n-1} + O((\log{M})^{n-2})\leq \tilde{S}_{\psi(\OO_K)}^{2}(M),
\end{equation}
where $c_K=\frac{h_K \pi^n 2 n^{n}}{(n-1)!}$ and $\tilde{N}_K = \frac{ \omega (n)^{n-1}}{R (n-1)!}(2^{-n}\sqrt{|d(K)|})$.
\end{proposition}
\begin{proof}
Following the proof of Proposition 4.4 in \cite{VLL2013} we have an estimate
$$
\tilde{S}_{\psi(\OO_K)}^{m}(M)\leq \zeta_K(1,\frac{M^{2n}}{n^n})(\tilde{N}_K(\log{M})^{n-1} + O((\log{M})^{n-2}),
$$
and the end result now follows from plugging  the result of Lemma \ref{tom} and Theorem \ref{Loxton} to this  estimate
\end{proof}
Again for principal ideal domains the upper bound is likely tight.

\section{Quasi-orthogonal codes from division algebras}\label{alamlike}

In the following we are considering the Alamouti-like multiblock codes from \cite{EK}. 
With respect to their complexity and other properties, all of the codes of this type are quasi-orthogonal. It is even possible to prove that many of the fully diverse quasi-orthogonal codes in the literature are unitarily equivalent to these multi-block codes. In the following we will use several 
results and concepts from the theory of central simple algebras. We  refer the reader to \cite{R} for an introduction to this theory.

Let us consider the field $E=KF$ that is a  compositum of  a  complex quadratic field $F$  and  a totally real Galois extension $K/\Q$ of degree $k$. We suppose  that $K \cap F=\Q$,  $Gal(F/\Q)= <\sigma>$ and $Gal(K/\Q)=\{\tau_1, \tau_2,\dots, \tau_k \}$. Here $\sigma$ is simply the complex conjugation. We can then write that $Gal(FK/\Q)=Gal(K/\Q)\otimes<\sigma>$.

Let us now consider  a cyclic division algebra
$$
\mathcal{D}=(E/K,\sigma,\gamma)=E\oplus uE,
$$
where   $u\in\mathcal{D}$ is an auxiliary generating element subject to the relations
$xu=ux^*$ for all $x\in E$ and $u^2=\gamma\in \OO_K$, where $()^*$ is the complex conjugation. 
We can consider $\D$ as a right  vector space over $E$ and every element   $a=x_1+ux_2 \in \D$ maps to
$$
\phi(a)=
\begin{pmatrix}
x_1& x_2\\
\gamma x_2^*&x_1^*
\end{pmatrix}.
$$
 
This mapping can then be extended into 
a multi-block representation $\phi:\D\mapsto M_{2k}(\C)$.
\begin{equation}
\psi(a)= \mathrm{diag}(\tau_1(\phi(a)),\tau_2(\phi(a))\dots, \tau_k(\phi(a))).
\end{equation}

\begin{exam}
In  the case where $k=2$ each element $a\in \D$ gets mapped as
$$
\psi(a)= 
\begin{pmatrix}
x_1& x_2& 0&0\\
\gamma x_2^*&x_1^*&0&0\\
0&0&\tau(x_1)&\tau(x_2)\\
0&0&\tau(\gamma x_2)^*&\tau(x_1)*\\
\end{pmatrix}.
$$
\end{exam}

In order to build a space-time lattice code from the division algebra $\D$ we will need  the following definition.

\begin{definition}
Let $\mathcal{O}_K$ be the ring of integers of $K$. An \emph{$\mathcal{O}_K$-order} $\Lambda$ in $\D$ is a subring of $\D$, having the same identity element as $\D$, and such that $\Lambda$ is a finitely generated module over $\mathcal{O}_K$ and generates $\D$ as a linear space over $K$. 
\end{definition}

Let us suppose that $\Lambda$ is an $\OO_K$-order  in $\D$. We call 
\emph{$\phi(\Lambda)$} an \emph{order code}. In the rest of this paper, we suppose that the division algebras under consideration are of the previous type.

\begin{lemma}\label{normindex}
If $\Lambda$ is an $\OO_K$-order in $\D$
\begin{equation}
|\det(\phi(x))|=\sqrt{[\Lambda:x\Lambda]},
\end{equation}
where $x$ is a non-zero element of $\Lambda$.
\end{lemma}

\begin{lemma}\label{perustulos}
Let us suppose that $\Lambda$ is  a $\OO_K$-order of a division algebra $\D$ with  center $K$ of degree $k$  and that  $\phi$ is a multi-block representation.
Then the order code $\phi(\Lambda)$ is a $4k$-dimensional  lattice in  the space $M_{2k}(\C)$ and
$$
\mindet{\psi(\Lambda)}=1.
$$
\end{lemma}

Let $\D$ be an index-$n$ $K$-central division algebra and $\Lambda$  a $\OO_K$-order in $\D$.
The  (right) \emph{Hey zeta function}  of the order $\Lambda$ is
$$
\zeta_{\Lambda}(s)=\sum_{I\in {\bf I}_{\Lambda}}\frac{1}{[\Lambda:I]^{s}},
$$
where $\Re(s)>1$ and ${\bf I}_{\Lambda}$ is the set of right ideals of $\Lambda$. When $\Re(s) >1$, this series is converging.

The unit group $\Lambda^*$ of an order $\Lambda$  consists of elements $x\in \Lambda$ such that there exists a $y\in \Lambda$ with 
$xy=1_{\A}$.  Another way to define this set is $\Lambda^*=\{x\in \Lambda \,|\, |\det{\psi(x)}|=1\}$.

\subsection{Inverse determinant sums of quasi-orthogonal codes}

Let us suppose that $K$, $\D$ and $\Lambda$ are as in the previous section and that $[K:\Q]=k$. We then have that
$\phi(\Lambda)$ is  a $4k$-dimensional NVD lattice in $M_{2k}(\C)$ and we can consider 
the growth of the sum
$$
\sum_{\psi(x)\in \psi(\Lambda)(M)}\frac{1}{|\det{\psi(x)}|^{2n_r}}=S^{2n_r}_{\psi(\Lambda)}(M).
$$

Just as in \cite{VLL2013} the previous sum can be analyzed further into
\begin{equation} \label{basicsum}
S^{2n_r}_{\psi(\Lambda)}(M)=\sum_{x\in X(M)}\frac{|\psi(x\Lambda^*) \cap B(M)|}{|\det(\psi(x))|^{2n_r}},
\end{equation}
where $X(M)$ is some collection of elements $x \in \Lambda$ such that $\norm{\psi(x)}_F\leq M$, each generating a different right ideal.

\subsection{ Uniform upper and lower bounds for $|\psi(x\Lambda^*)\cap B(M)|$}
The key element in the analysis  of $|\psi(x\Lambda^*)\cap B(M)|$ is the following.

\begin{lemma}[Eichler]
The unit group $\Lambda^*$ has a subgroup
$$
\OO_K^*=\{x \,|\,x\in \Lambda^*,  x\in \OO_K \},
$$
and we have $[\Lambda^*:\OO_K^*]<\infty$.
\end{lemma}

Let $j=[\Lambda^*:\OO_K]$. By choosing a set $\{a_1,\ldots,a_j\}$ of coset leaders of $\OO_K^*$ in $\Lambda^*$, we   have  that
\begin{equation}\label{}
 |\psi(x\Lambda^*)\cap B(M)|\leq \sum_{i=1}^j |\psi(xa_i\OO_K^*)\cap B(M)|.
\end{equation}
In order to give an uniform upper bound for $|\psi(x\Lambda^*)\cap B(M)|$, it is now enough to give a uniform upper bound for $|\psi(xa_i\OO_K^*)\cap B(M)|$. Before stating our main results we need few lemmas. We will skip the proofs of some of them.

\begin{lemma}\label{basic}
Let us suppose that $A$ is a diagonal matrix in $M_n(\C)$ with $|\det{A}|\geq 1$. We then have that
$$
|A\psi(\OO_K^*)\cap B(M)|\leq |\psi(\OO_K^*)\cap B(cM)|,
$$
where $c$ is a fixed constant, independent of $A$ and $M$.
\end{lemma}

\begin{lemma}\label{approx}
Let us suppose that  $x$ and $y$ are elements in $\OO_{KF}$, we then have that
$$
|\psi(x)\psi(\OO_K^*)\cap B(M)|\leq |\psi(\OO_K^*)\cap B(cM)|,
$$
and
$$
|\psi(uy)\psi(\OO_K^*)\cap B(M)|\leq |\psi(\OO_K^*)\cap B(cM)|,
$$
where $c$ is a real constant independent of $x, y$ and $M$.
\end{lemma}
\begin{proof}
The first result is simply Lemma \ref{basic} and the second follows as $\psi(u)$ is a fixed matrix.
\end{proof}

\begin{lemma}\label{orto}
Let us suppose that $x$ and  $y$ are elements in $E$. We then have that
$$
||\psi(x)+\psi(uy)||_F^2 = ||\psi(x)||_F^2+||\psi(uy)||_F^2.
$$

\end{lemma}
\begin{proof}
By an elementary calculation we see that $<\psi(x),\psi(uy)>=0$ and the claim follows.
\end{proof}

\begin{proposition}\label{mainhelp}
Let us suppose that $x\in \Lambda$, we then have that
$$
|\psi(x)\psi(\OO_K^*)\cap B(M)|\leq |\psi(\OO_K^*)\cap B(cM)|,
$$
where $c$ is a constant independent of $M$ and $x$.
\end{proposition}
\begin{proof}
Let us suppose first that  $x= x_1+ux_2$, where $x_i\in \OO_{E}$ and where $u^2\in \OO_K$.
According to Lemma \ref{orto}, we have that
$$
||\psi(x)\psi(y)||^2= ||\psi(x_1)\psi(y)||^2 + ||\psi(ux_2)\psi(y)||^2,
$$
for any  $y \in \OO_{E}$.

Therefore if $\psi(x)\psi(y)\in B(M)$, then also
$$
\psi(x_1)\psi(y) \in \, B(M) \, \mathrm{and} \, \psi(ux_2)\psi(y) \, \in B(M).
$$

It follows that we can upper bound $|\psi(x)\psi(\OO_K^*)\cap B(M)|$ with
 $$
 \mathrm{max}\{|\psi(x_1)\psi(\OO_K^*) \cap B(M)|, |\psi(ux_2)\psi(\OO_K^*) \cap B(M)| \}.
$$

According to lemma \ref{approx} we then have that
$$
|\psi(x)\psi(\OO_K^*)\cap B(M)|\leq|\psi(\OO_K^*)\cap B(cM)|.
$$
Let us now suppose that $\Lambda$ is a general order in $\D$. As $\psi(\Lambda)$ is 
finitely generated as an additive group in $M_n(E)$, we can choose an integer
$d$ such that $d\psi(\Lambda) \subseteq \psi(\OO_{E})+\psi(u\OO_{E})$. 
The result now follows from the previous consideration.
\end{proof}

\begin{proposition}\label{uniformmain}
Using the previous notation we have
\begin{eqnarray*}
|\psi(x\Lambda^*)\cap B(M)| \leq [\Lambda^*:\OO_K] \cdot\\
\log{M}^{k-1} \frac{ \omega (k)^{k-1}}{R (k-1)!}+O(\log{M}^{k-2}).
\end{eqnarray*}
\end{proposition}

\begin{proof}
Let $j=[\Lambda^*:\OO_K]$. By choosing a set $\{a_1,\ldots,a_j\}$ of coset leaders of $\OO_K^*$ in $\Lambda^*$, we then  have  that
$$
 |\psi(x\Lambda^*)\cap B(M)|\leq \sum_{i=1}^j |\psi(xa_i\OO_K^*)\cap B(M)|.
$$
According to Proposition  \ref{mainhelp} we then have that
$$
|\psi(x\Lambda^*)\cap B(M)| \leq [\Lambda^*:\OO_K]||\psi(\OO_K^*)\cap B(cM)|.
$$
Applying  Theorem \ref{Loxton} to this equation, we get the final result.
\end{proof}

\subsection{Upper and lower bounds for inverse determinant sums of quasi-orthogonal codes}

\begin{proposition}\label{}
Let us suppose that $[K:\Q]=k$ and set $n=2k$. We then have that $\psi(\Lambda)$ is a $2n$-dimensional lattice in $M_{n}(\C)$ and
$$
\log{M}^{n/2-1} \frac{ \omega {(\frac{n}{2})}^{n/2-1}}{R (n/2-1)!}+O(\log{M}^{n/2-2})\leq S^{2n_r}_{\psi(\Lambda)}(M) 
$$
$$
\leq\zeta_{\Lambda}(n_r)[\Lambda^*:\OO_K] \log{M}^{\frac{n-2}{2}} \frac{ \omega (\frac{n}{2})^{\frac{n-2}{2}}}{R (\frac{n-2}{2})!}+O(\log{M}^{n/2-2}),
$$
where $n_r>1$ and $R$ and $\omega$ are the regulator and the number of roots of unity in the center $K$.
\end{proposition}
\begin{proof}
As previously mentioned, we can imitate \cite{VLL2013} to get
\begin{equation}\label{start}
S^{2n_r}_{\psi(\Lambda)}(M)=
\sum_{x\in X(M)}\frac{|\psi(x\Lambda^*) \cap B(M)|}{|\det(\psi(x))|^{2n_r}}.
\end{equation}
According to Lemma \ref{normindex} we have that  $\abs{\det(\psi(x))}^{2n_r}=
[\Lambda:x\Lambda]^{n_r}$.
Now 
$$
\sum_{x \in X(M)} \frac{1}{\abs{\det(\psi(x))}^{2n_r}} \leq \sum_{x \in X(M)} \frac{1}{[\Lambda:x\Lambda]^{n_r}}
 \leq \zeta_{\Lambda}(n_r).
$$
Applying this inequality  with Proposition \ref{uniformmain} to \eqref{start} now gives us the final result.
\end{proof}


\section{Quasi-orthogonal codes are better than  diagonal number field codes}
Let us now suppose we have an $n\times n_r$-MIMO channel, (for simplicity we assume $n_r>1$). For the existence of quasi-orthogonal code we also have to assume that $2\mid n$. Let us  now compare the growth of determinant sums of quasi-orthogonal and comparable diagonal number field codes in this  $n\times n_r$-MIMO channel.

In order to build a quasi-orthogonal code  $\psi(\Lambda)$ in $M_n(\C)$ the center $K$ of the algebra $\D$ must be  an $n/2$-dimensional totally real number field. For  a number field code $\psi(\OO_L)\subseteq M_n(\C)$, the field $L$ must be  an $n$-dimensional extension of some complex quadratic field $F$.

As we earlier saw, we have that
$$
 \sum_{X\in \psi(\Lambda)(M)} \frac{1}{|\det(X)|^{2n_r}}  = \theta(|\psi(\Lambda^*)\cap B(M)|)
$$
and
$$
|\psi(\Lambda^*)\cap B(M)|=\theta(|\psi(\OO_K^*)\cap B(M)|)=\theta(\log{M}^{n/2-1}).
$$
Therefore
$$
 \sum_{X\in \psi(\Lambda)(M)} \frac{1}{|\det(X)|^{2n_r}}  = \theta(\log{M}^{n/2-1}).
$$

On the other hand for the number field code we have that
\begin{eqnarray*}
 \sum_{X\in \psi(\OO_L)(M)} \frac{1}{|\det(X)|^{2n_r}} &= \theta(|\psi(\OO_L^*)\cap B(M)|)\\
 &=\theta(\log{M}^{n-1}).
\end{eqnarray*}
Here the last result follows from \cite[Theorem 2]{EL}.

We can now see that the growth of the inverse determinant sum for the quasi-orthogonal code is considerably lower than that 
of the number field code. This is due to the fact that the unit group of the order $\Lambda$ is essentially that of a low degree real number field.
We note that this difference can not be captured in the context of DMT as both of these codes have the same DMT curve.

\section*{Acknowledgement}
 The research of  R. Vehkalahti is supported  by the Academy of
Finland  grant  \#252457.

\end{document}